\documentclass[10pt, conference, letterpaper]{IEEEtran}

\usepackage{booktabs} 

\usepackage{algorithm}
\usepackage{algpseudocode}
\usepackage[fleqn]{amsmath}
\usepackage{amsthm}
\usepackage{amssymb}
\usepackage{verbatim}
\usepackage{graphicx}
\usepackage{paralist}
\usepackage{color}

\usepackage{xcolor,colortbl}
\usepackage{epstopdf}  
\usepackage{subfigure}
\usepackage{cite}
\usepackage{setspace}
\usepackage{enumitem}
\usepackage{flushend}

\usepackage{caption}
\usepackage{accents}

\newtheorem{lemma}{Lemma}
\newtheorem{theorem}{Theorem}

\newtheorem{problem}{Problem}
\newtheorem{definition}{Definition}

\newtheorem{example}{Example}

\floatname{algorithm}{Algorithm}

\usepackage{booktabs}

\usepackage{enumitem}
\setitemize{itemsep=0pt}

\allowdisplaybreaks

\title{Sequential Resource Access: Theory and Algorithm}

\author{
    \IEEEauthorblockN{Lin Chen\IEEEauthorrefmark{1}\IEEEauthorrefmark{5}, Anastasios Giovanidis\IEEEauthorrefmark{2}, Wei Wang\IEEEauthorrefmark{3}, Lin Shan\IEEEauthorrefmark{4}}
    \IEEEauthorblockA{\IEEEauthorrefmark{1}School of Computer Science and Engineering, Sun Yat-sen University, Guangzhou, China, chenlin69@mail.sysu.edu.cn}
    \IEEEauthorblockA{\IEEEauthorrefmark{2}Sorbonne University, CNRS-LIP6, Paris, France, Anastasios.Giovanidis@lip6.fr}
    \IEEEauthorblockA{\IEEEauthorrefmark{3}College of Information Science and Electronic Engineering, Zhejiang University, Hangzhou, China, wangw@zju.edu.cn}
    \IEEEauthorblockA{\IEEEauthorrefmark{4}Dept. of Electrical and Computer Engineering, Stony Brook University, Stony Brook, NY, USA, shan.x.lin@stonybrook.edu}
    \IEEEauthorblockA{\IEEEauthorrefmark{5}Key Laboratory of Machine Intelligence and Advanced Computing, Ministry of Education, China}
}

\begin{document}

\maketitle

\begin{abstract}
We formulate and analyze a generic sequential resource access problem arising in a variety of engineering fields, where a user disposes a number of heterogeneous computing, communication, or storage resources, each characterized by the probability of successfully executing the user's task and the related access delay and cost, and seeks an optimal access strategy to maximize her utility within a given time horizon, defined as the expected reward minus the access cost. We develop an algorithmic framework on the (near-)optimal sequential resource access strategy. We first prove that the problem of finding an optimal strategy is \textsc{NP}-hard in general. Given the hardness result, we present a greedy strategy implementable in linear time, and establish the closed-form sufficient condition for its optimality. We then develop a series of polynomial-time approximation algorithms achieving $(\epsilon,\delta)$-optimality, with the key component being a pruning process eliminating dominated strategies and, thus maintaining polynomial time and space overhead.
\end{abstract}

\section{Introduction}

We consider the following generic resource access problem: a user needs to execute a communication or computing task; there are a set of resources she may access; by accessing resource $i$, she can successfully execute her task with probability $p_i$ and in that case obtains a unit reward; accessing resource $i$ incurs a cost $c_i$ and delay $d_i$; the user seeks a \textit{sequential} resource access strategy maximizing her expected utility, defined as the expected reward minus the cost, within a time horizon $T$ corresponding to the maximal delay she disposes to complete her task.

The above resource access problem arises in a variety of engineering fields, where a decision maker disposes a number of heterogeneous computing, communication, or storage resources, and needs to find an optimal strategy to access them to execute her task. Examples fitting into this formulation include communication and computing task offloading, cached data access, opportunistic forwarding, and user-centric network selection (cf. Section~\ref{sec:formulation} for a detailed description). Theoretically, the intrinsic problem structure we attempt to capture in our formulation is the sequential selection of resources with a hard time constraint. We believe that such formulation, despite being generic, can provide valuable insights in many emerging networking, communication, and computing scenarios.

Furthermore, application examples of our formulation extend well beyond the realm of communication and networking, and include a class of sequential searching problems that can be formulated by the following intuitive example. Given a number of places, each having a certain probability of containing a prize (e.g., a fugitive actively searched by the police), an agent aims at finding the prize by searching the places. Searching a place incurs certain cost and delay. The agent seeks an optimal sequence of places to search in order to maximize the overall probability of finding the prize while limiting the searching cost within a given time horizon.

Motivated by the above observation, we embark in this paper on a systematic analysis of the sequential resource access problem in its generic form. By generic we mean that no specific problem or context is assumed for the analysis. The only important assumption, made mainly for
mathematical tractability, is the independence of success probability among resources. While limiting the applicability of our analysis to some extent, this assumption is justified in many
practical situations, where resources are independent one to the other. In case where resources are correlated among them, our model can be extended by taking into the account such inter-resource correlation by e.g., forming super-resources representing correlated ones. We leave a detailed analysis of the correlated case for future research. Despite the generic nature of our analysis, the algorithms we develop are readily implementable once instantiated with system parameters.  

Our major technical contribution in this paper is a comprehensive algorithmic framework deriving the optimal or near-optimal sequential resource access strategy. To this end, we first analyze the homogeneous case where the access time $d_i$ is the same among resources. We present a greedy strategy that can be implemented in linear time and mathematically establish sufficient conditions for its optimality. We then develop a polynomial-time algorithm that gives an optimal strategy and design a preprocessing procedure to further reduce its complexity. We then turn to the heterogeneous case where the access time may be different among resources. We prove that the problem of finding an optimal strategy is \textsf{NP}-hard. Given its hardness, we develop a series of polynomial-time approximation algorithms approaching an optimum solution in the sense of $(\epsilon,\delta)$-optimality. 

The paper is organized as follows. In Section~\ref{sec:formulation} we formalize the sequential resource access strategy and present some structural properties. In Sections~\ref{sec:homogeneous} and~\ref{sec:heterogeneous}, we develop a set of greedy and approximation strategies for the homogeneous and heterogeneous cases, respectively. 
Section~\ref{sec:simu} presents the simulation results. We review related work in Section~\ref{sec:related_work}. Section~\ref{sec:conclusion} concludes the paper. 

\section{Model and Problem Formulation}
\label{sec:formulation}

We consider the sequential resource access problem formulated in the Introduction. The user disposes a set $\cal N$ of $N$ resources; each resource $i$ is characterized by a triple $(p_i,c_i,d_i)$, known to the user, corresponding to the success probability, access cost, and access delay of the resource. Our goal is to design an optimal resource access strategy, concisely termed as \textit{strategy}, that maximizes the utility of the user within a time horizon $T$.

\subsection{Assumptions}

The formulation of our problem hinges on the following two assumptions.

First, resources are mutually independent. This assumption is justified in many practical situations. For example, resources may map to distinct communication channels where the probability of  accomplishing a transmission in one channel is independent to that of the other; resources may also map to a set of hard disks where their failure probabilities are mutually independent. In case of correlated resources, our model can be extended by taking into the account such inter-resource correlation by, e.g., forming super-resources representing correlated ones.

Second, resources are selected without replacement. That is, we focus on the case where the probability $p_i$ of each resource $\pi_i$ represents the uncertainty about its availability, which is revealed once $\pi_i$ is accessed. This assumption is valid when the environment is relatively stable during the access process, e.g., in the slow fading case, where each resource maps to a wireless channel, with the channel coherence time longer than the time horizon $T$.

\begin{table}
\begin{tabular}{p{0.6cm} p{7.35cm} }
\toprule
& Section~\ref{sec:formulation}: model and problem formulation \\
\toprule
$\cal N$ & Set of accessible resources, $N=|\cal N|$ \\
$p_i$ & Success prob. of resource $i$ \\
$q_i$ & $1-p_i$ \\
$c_i$ & Access cost of resource $i$ \\
$d_i$ & Access delay of resource $i$ \\
$T$ & Maximal delay bound \\
$\pi$ & Resource access strategy, $\pi=(\pi_1,\cdots,\pi_n)$ \\
$n$ & Number of resources in $\pi$, $n=|\pi|$ \\
$d(\pi)$ & Access delay of strategy $\pi$, $d(\pi)=\sum_{i=1}^{n} d_{\pi_i}$ \\
$\pi^*$ & Optimal feasible resource access strategy \\
$\Pi$ & Set of all resource access strategies \\
$\Pi_f$ & Set of all feasible resource access strategies \\
$U(\pi)$ & Utility function \\
$V(\pi)$ & Pseudo-cost function, $V(\pi)=1-U(\pi)$ \\
\toprule
& Section~\ref{sec:homogeneous}: the homogeneous case \\
\toprule
$y_i$ & Rank of $q_i$ among all $q$'s \\
$z_i$ & Rank of $c_i$ among all $c$'s \\
$\widehat{\pi}$ & Resource access strategy output by our algorithms \\
\toprule
& Section~\ref{sec:heterogeneous}: the heterogeneous case \\
\toprule
$c_{min}$ & $\min_{i\in {\cal N}} c_i$ \\
$d_{min}$ & $\min_{i\in {\cal N}} d_i$, normalized to $1$ \\
$\widehat{d}_i$ & Discretized access delay \\
$\widehat{V}$ & Discretized pseudo-cost \\
$\widehat{d}(\pi)$ & Discretized access delay of $\pi$, $\widehat{d}(\pi)=\sum_{i=1}^{|\pi|} \widehat{d}_{\pi_i}$ \\
$\lambda$ & Scaling parameter corresponding to access time \\
$\mu$ & Scaling parameter corresponding to pseudo-cost \\
\bottomrule
\end{tabular}
\label{tab:notation}
\caption{Main notations}
\end{table}

\subsection{Sequential Resource Access Problem Formulation}

Mathematically, a strategy $\pi$ can be defined as an ordered set of resources $\{\pi_1,\cdots,\pi_n)$, $\pi_i\in{\cal N}$ for $1\le i\le n$, that the user accesses sequentially. We denote the number of resources accessed in $\pi$ by $n$. As resources are selected without replacement, we have $n\le N$. We are particularly interested in the case where $n$ is significantly smaller than $N$, thus requiring a carefully designed access strategy. 

We define the following function $U({\pi})$ to denote the user's utility as a function of her strategy $\pi$. 
$$U(\pi)\triangleq R(\pi)-C(\pi).$$
Specifically, $U({\pi})$ is defined as the expected reward minus the cost.
$R(\pi)\triangleq 1-\prod_{i=1}^n (1-p_{\pi_i})$ denotes the expected normalized reward, where the user gets a unit reward (e.g., one euro) if she successfully executes her task before the deadline.
\begin{multline*}C(\pi)\triangleq \sum_{i=1}^n \left[\prod_{j=1}^{i-1} (1-p_{\pi_j})\right] p_{\pi_i} \left(\sum_{k=1}^i c_{\pi_k}\right) \\ 
+\left[\prod_{j=1}^n (1-p_{\pi_j})\right]\left(\sum_{k=1}^n c_{\pi_k}\right)
\end{multline*}
denotes the expected normalized cost (e.g., also in euro), where $\left[\prod_{j=1}^{i-1} (1-p_{\pi_j})\right] p_{\pi_i}$ is the probability that the user fails in executing her task with the resources $\pi_1$ to $\pi_{i-1}$ and succeeds with resource $i$ with $\sum_{k=1}^i c_{\pi_k}$ being the related cost, $\prod_{j=1}^n (1-p_{\pi_j})$ is the probability of failure with all the $n$ resources in $\pi$ with $\sum_{k=1}^n c_{\pi_k}$ being the related cost. 

For ease of
presentation, we define a pseudo-cost function
$$V(\pi) \triangleq 1-U(\pi).$$
The problem of maximizing the utility function $U(\pi)$ maps to minimizing the pseudo-cost function $V(\pi)$. We next give the definition of feasible strategy and formulate the sequential resource access problem. 

\begin{definition}[Feasible strategy]
For any strategy $\pi\triangleq\{\pi_i\}_{i=1}^n$, let $d(\pi)\triangleq \sum_{i=1}^n d_{\pi_i}$ denote the access delay of $\pi$. We call $\pi$ feasible if its access delay does not exceeds $T$, i.e., $d(\pi)\le T$. We denote $\Pi_f$ the set of all feasible strategies.
\end{definition}

\begin{example}
The user disposes $100$ resources, half with access delay $d_i=1$, termed as type-I resources, half with $d_i=2$, termed as type-II resources, and finite-horizon $T=3$. The feasible strategy set $\Pi_f$ consists of the following categories of strategies: (1) accessing $1$, $2$, or $3$ type-I resource(s), (2) accessing $1$ type-II resource, (3) accessing $1$ resource per type.
\end{example}

\begin{problem}[Sequential resource access]
\label{def:pb}
The sequential resource access problem seeks a feasible strategy $\pi$ minimizing the pseudo-cost function $V(\pi)$, i.e., $\min_{\pi\in \Pi_f} V(\pi)$.
\end{problem}

Our problem is by nature a combinatorial optimization problem. Consider a degenerated case with zero access cost. The problem becomes purely combinatorial and can be algebraically formulated as the following Knapsack problem.
$$\min_{\pi} \sum_{i\in\pi} \log(1-p_i),$$
subject to the Knapsack constraint $\min_{\pi} \sum_{i\in\pi} d_i \le T$. The order of accessing resources does not matter in this case. On the other hand, in the generic setting, the user may tend to access more reliable resources with high $p_i$'s; these resources may also incur larger access delay and cost; therefore, she needs to strike a balance among the success probability, access cost, and delay. Moreover, the order of the accessed resources is also important and needs to be optimized.

\subsection{Applicability of Our Problem Formulation}

Our generic formulation of the sequential resource access problem is readily applicable in a wide range of resource access problems in emerging communication and computing applications. Below we give three concrete examples.

\textbf{Data Access in Network Caching}.
Caching is widely deployed in emerging networking systems such as 5G edge and in-network caches~\cite{6736746}, content delivery networks (CDN)~\cite{201471}. In network caching systems, a cache can be regarded as a data store holding a subset of data that may be accessed by users. 
Intuitively, knowing which item is stored in each data store can significantly improve user experience. However, maintaining such information may be too expensive. As an alternative solution, it is more practical to maintain an approximate catalog of items at each data store based on compact data structures such as Bloom filters~\cite{10.1145/362686.362692}. The price to pay for the space compactness is the well-known false positive, where a Bloom filter returns a positive response while the corresponding data item is not stored at the data store. The false positive rate of a Bloom filter of size $m$ and $k$ hash functions storing $n$ data items is $\left[1-(1-1/m)^{kn}\right]^k$, approximately $\left(1-e^{-kn/m}\right)^k$. 

Now consider the the situation where a user needs to fetch a data item. There are a number of candidate data stores returning positive responses after checking the corresponding Bloom filters. Accessing each data store $i$ incurs a cost $c_i$ and delay $d_i$. The user disposes time $T$ to fetch the data item. She gets a reward $r_1$ if successfully fetching it by the deadline, and pays a penalty $r_2$ otherwise. We can formulate two optimization problems faced by the user. In the first problem, investigated in~\cite{8737427}, the user can access multiple data stores and seeks an optimal set of data stores to access simultaneously. The second problem captures the situation where the user is limited to access one cache each time, but can perform multiple queries sequentially as long as the total delay does not exceed $T$. Compared to accessing multiple caches simultaneously, sequential access can reduce the user's total access cost. The sequential data access problem can be formulated by our sequential resource access problem by mapping the false positive rate of the Bloom filter of cache $i$ to $1-p_i$, and normalizing the user's reward, penalty, and cost. 

\textbf{Interest Forwarding in Information Centric Networks (ICN)}.
In ICNs~\cite{ZHANG20133128}, a content is typically divided into chunks. Each chunk is addressed by a unique ID and may have many identical cached copies in the ICN routers across the Internet. A chunk is located and requested by forwarding the so-called \textit{interests}. A user can forward her interest to one or more neighbor ICN routers. If there is no bandwidth or other cost limitation, the user can forward her interest to all available neighbors. However, if there is a bandwidth limitation, or the user has to pay for the interest or delivered content, then she needs to carefully choose which neighbors to forward her interest and in what order, rather than simple flooding~\cite{6486116}. By mapping the neighbor set to $\cal N$, the probability that neighbor $i$ can return a copy of the searched chunk to $p_i$, the average response delay and the access cost to $d_i$ and $c_i$ respectively, the ICN interest forwarding problem can be cast to the sequential resource access problem. 

\textbf{Opportunistic Packet Forwarding}.
In opportunistic packet forwarding, a user needs to decide the sequence of invoking a set of potential forwarders to transmit a data packet. Each forwarder $i$ has a certain probability $p_i$ of successfully executing the transmission task depending on its channel condition, which is assumed to be stable during the considered time horizon. Invoking forwarder $i$ incurs a cost $c_i$ (e.g., in terms of energy consumption) and delay $d_i$. The user seeks an optimal sequential strategy to invoke a subset of forwarders to maximize the packet delivery rate within the tolerable delay $T$ by taking into account the related cost. This optimal opportunistic forwarding problem also fits in our formulation.

\subsection{Structural Properties of Optimal Feasible Strategy}

We conclude this section by showing the following structural
properties of $V(\pi)$ and any optimal feasible strategy $\pi^*$. To streamline our presentation and due to page limit, readers are referred to the anonymous technical report~\cite{techreport} for the proof of all the lemmas and theorems.

\begin{lemma}[Structural properties of $V(\pi)$]
\label{lemma:v}
Define $q_i\triangleq 1-p_i$, $\forall i\in{\cal N}$. The following properties hold\footnote{To make the notation concise, we denote $\prod_{j=1}^i q_{\pi_j}=1$ for $i<1$.}:
\begin{enumerate}
    \item $\displaystyle V(\pi)=\prod_{i=1}^n q_{\pi_i} + \sum_{i=1}^n \left(\prod_{j=1}^{i-1}q_{\pi_j}\right)c_{\pi_i}$;
    \item If $q_{\pi_i}+c_{\pi_i}\le 1$ for $1\le i\le n$, then $V(\pi)\le 1$;
    \item Given any strategy $\pi$, if there exists a resource $\pi_k\in\pi$ such that $q_{\pi_k}+c_{\pi_k}>1$, it holds that $V(\pi)>V(\pi_{-k})$, where $\pi_{-k}$ denotes the strategy by removing $\pi_k$ from $\pi$.
\end{enumerate}
\end{lemma}

The first property of Lemma~\ref{lemma:v} demonstrates that  $V(\pi)$ consists of two terms: (1) the expected normalized loss due to failing to execute the task $\prod_{i=1}^n q_{\pi_i}$, (2) the expected aggregated access cost $\sum_{i=1}^n \left(\prod_{j=1}^{i-1}q_{\pi_j}\right)c_{\pi_i}$. The second property can be explained intuitively. Upon accessing each resource $\pi_i$, the user gets an expected reward $p_{\pi_i}$, with an incurred access cost $c_{\pi_i}$. If the expected reward outweighs the cost for each resource $\pi_i$, i.e., $p_{\pi_i}\ge c_{\pi_i}$, the global utility $U(\pi)$ is logically non-negative. It follows that $V(\pi)=1-U(\pi)\le 1$. Note that $p_{\pi_i}=1-q_{\pi_i}$, a sufficient condition to achieve $V(\pi)\le 1$ is $q_{\pi_i}+ c_{\pi_i}\le 1$, as stated in the second part of Lemma~\ref{lemma:v}. The third property demonstrates that any rational user never accesses any resource $i$ where $q_i+c_i>1$, as simply removing it improves the performance. We can thus safely focus on the case where $q_i+c_i\le 1$ holds for $1\le i\le N$.

\begin{lemma}[Structural properties of optimal strategy $\pi^*$]
Let $\pi^*$ denote an optimal strategy, it holds that $p_{\pi_i^*}/c_{\pi_i^*}$ is non-increasing in $i$, i.e., $p_{\pi_i^*}/c_{\pi_i^*}\ge p_{\pi_{i+1}^*}/c_{\pi_{i+1}^*}$, $1\le i\le |\pi^*|-1$.
\label{lemma:optimal_property}
\end{lemma}

Lemma~\ref{lemma:optimal_property} demonstrates that, once the set of resources to access at $\pi^*$ is determined, it suffices to access them in decreasing order of $p_i/c_i$. In other words, our problem can be transformed to finding an optimal set of resources to access, which is by nature a combinatorial optimization problem.

\section{The Homogeneous Case}
\label{sec:homogeneous}

This section focuses on the homogeneous case where the access delay $d_i$ is identical among resources. To make the notation concise without losing generality, we set $d_i=1, \forall i\in{\cal N}$, i.e., the user can access at most $T$ resources. Our motivation of starting with the homogeneous case is to tackle the problem in a progressive way. As we will demonstrate in this section, analyzing the homogeneous case allows us to obtain more insights on the structure of an optimal strategy, which are useful in the study of the heterogeneous case. We emphasize that the homogeneity only concerns the access delay. Other parameters, such as success probability and access cost, may still be heterogeneous among resources. 

\subsection{A Greedy Strategy Based on $p_i/c_i$}

We first investigate a greedy strategy sequentially accessing $T$ resources\footnote{By sequentially accessing a set of resources, we mean by accessing sequentially those resources until a success or the end of time horizon $T$.} 
in the decreasing order of $p_i/c_i$, as formalized in Definition~\ref{def:greedy}. Our greedy strategy is motivated by the structural property of an optimal strategy in Lemma~\ref{lemma:optimal_property}. From an economic angle, the greedy strategy sequentially chooses $T$ resources by decreasing ratio of success probability (representing profit in certain sense) to access cost. To make our analysis concise, we assume that for any pair of resources $i$ and $j$, $c_i\ne c_j$.\footnote{In case of tie where $c_i=c_j$, we add a small quantity $\epsilon$ to either $c_i$ or $c_j$ to break the tie, which leads to at most $\epsilon$ in the utility of the optimal strategy. Our analysis can thus be extended in the generic case.}

\begin{definition}[Greedy strategy]
The greedy strategy consists of sequentially accessing $T$ resources by decreasing ratio $p_i/c_i$, in case of tie choosing the resources by increasing $c_i$. Mathematically, by sorting the resources such that for each $1\le i\le N-1$ either (1) ${p_i}/{c_i}>{p_{i+1}}/{c_{i+1}}$ or (2) ${p_i}/{c_i}={p_{i+1}}/{c_{i+1}}$ and $c_i<c_{i+1}$, the greedy strategy sequentially accesses the first $T$ resources.
\label{def:greedy}
\end{definition}

In Theorem~\ref{theorem:optimality_greedy}, we establish the sufficient condition under which the greedy strategy is optimal. In the sequel analysis we assume that the resources are sorted according to Definition~\ref{def:greedy}.

\begin{theorem}[Optimality condition of greedy strategy]
If $c_i< c_{i+1}$ and $\frac{p_i-p_{i+1}}{c_i-c_{i+1}}< 1$ hold for $1\le i\le N-1$, then the greedy strategy is the only optimal strategy.
\label{theorem:optimality_greedy}
\end{theorem}

By treating $p_i$ as a function of  $c_i$, the condition $\frac{p_i-p_{i+1}}{c_i-c_{i+1}}< 1$ can be essentially regarded as the discrete version of the economic property on the marginal utility $\Delta p/\Delta c$. 
What we essentially demonstrate is that, if $p_i$ does not increase as much as $c_i$ among resources, then the greedy strategy is optimal.

Despite our efforts in characterizing the optimality of the greedy strategy, the optimality condition established in Theorem~\ref{theorem:optimality_greedy} may be too stringent in many cases, and the greedy strategy may be far from optimal. To illustrate this, we consider an example where $N=2$, $p_1=0.2$, $c_1=0.1$, $p_2=0.9$, $c_2=0.5$, and $T=1$; clearly the greedy strategy accesses resource $1$ leading to utility $0.1$; however, the optimal strategy is to access resource $2$ leading to utility $0.4$.

Motivated by the above analysis, we proceed to derive the optimal strategy in the generic case.

\subsection{Optimal Strategy based on Dynamic Programming}


As in the previous subsection, we sort resources by decreasing $p_i/c_i$. By Lemma~\ref{lemma:optimal_property}, an optimal strategy corresponds to accessing a subset of resources in that order. It remains to find the subset. To this end, we define an auxiliary function $\widehat{V}(r,l)$ ($0\le r\le T$, $r\le l\le N$) to denote the minimal expected pseudo-cost by accessing at most $r$ among the first $l$ resources, i.e., $\widehat{V}(r,l)\triangleq \min_{\pi\subseteq {\cal N}_l,|\pi|=r} V(\pi)$, where ${\cal N}_l$ denotes the subset of $\cal N$ containing the first $l$ resources.

The optimal algorithm we develop is based on dynamic programming, hinging on the following recursive property:
\begin{multline}
\widehat{V}(r,l)=\min\{c_l+q_l\cdot \widehat{V}(r-1,l-1),\widehat{V}(r,l-1)\}, \\ 
1\le r\le T, \ r\le l\le N.
\label{eq:dp}
\end{multline}
The above property follows from the observation that $\widehat{V}(r,l)$ is the minimum of the following two strategies: (1) an optimal strategy that accesses $(r-1)$ resources among the first $(l-1)$ resources followed by accessing resource $l$, (2) an optimal strategy that accesses $r$ resources among the first $(l-1)$ resources without accessing resource $l$. The border values are given by $\widehat{V}(0,l)=1$, $0\le l\le N$.

The pseudo-code of our optimal algorithm is described in Algorithm~\ref{alg:baseline}, essentially consisting of two iterations (except the first \textbf{for} iteration for initialization). The first iteration computes the values of $\widehat{V}(r,l)$. Note that we only compute $\widehat{V}(r,l)$ for $r=1$ to $T$ and $l=r$ to $(N-T+r)$ for a given $r$, because these are sufficient to calculate the optimal cost $\widehat{V}(T,N)$. This can be physically implemented by a $T\times (N-T+1)$ array. The second iteration starts from $(T,N)$ and traces back to $(1,1)$ to derive the optimal strategy $\pi^*$. During this process, $r$ traces the current index of the resource accessed in $\pi^*$. Both the time and space complexity of Algorithm~\ref{alg:baseline} is $O(NT)$, or more precisely, $O((N-T)T)$.

\begin{algorithm}
\caption{Finding optimal strategy: homogeneous case}
\label{alg:baseline}
\begin{algorithmic}[1]
\State \textbf{Input}: $T$, $\{q_i,c_i\}_{1\le i \le N}$
\State \textbf{Output}: an optimal strategy $\pi^*=\{\pi^*_r\}_{1\le r\le T}$

\vspace{0.5em}
\State Sort resources by decreasing $p_i/c_i$
\For{$l=1$ \textbf{to} $N$}
    \State $\widehat{V}(0,l)\leftarrow 1$
\EndFor

\vspace{0.5em}
\For{$r=1$ \textbf{to} $T$, $l=r$ \textbf{to} $N-T+r$}
    \If{$c_l+q_l\cdot \widehat{V}(r-1,l-1)\le \widehat{V}(r,l-1)$}
        \State $\widehat{V}(r,l)\leftarrow c_l+q_l\cdot \widehat{V}(r-1,l-1)$
    \Else
        \State $\widehat{V}(r,l)\leftarrow \widehat{V}(r,l-1)$
    \EndIf
\EndFor

\vspace{0.5em}
\State $r\leftarrow T$
\For{$l=N$ \textbf{to} $1$}
    \If{$c_l+q_l\cdot \widehat{V}(r-1,l-1)\le \widehat{V}(r,l-1)$}
        \State $\pi_r^*\leftarrow l$
        \State $r\leftarrow r-1$
    \EndIf
\EndFor
\end{algorithmic}
\end{algorithm}

Finding an optimum strategy can be cast to finding a path between the root of a tree to a leaf. This formulation helps us gain more insights on the problem. Specifically, we build a graph $G\triangleq ({\cal V}, {\cal E})$, where the set of vertexes ${\cal V}\triangleq \{(r,l)\}$, $0\le r\le T, r\le l\le N$. We add an edge between the vertexes $(r,l)$ and $(r-1,l-1)$ if $c_l+q_l\cdot \widehat{V}(r-1,l-1)\le \widehat{V}(r,l-1)$, and between vertexes $(r,l)$ and $(r,l-1)$ otherwise, where $\widehat{V}(r,l)$ is derived recursively by~\eqref{eq:dp}. We can check that $G$ is a tree rooted at $\widehat{V}(0,0)$, and that $\widehat{V}(T,N)$ is a leaf of $G$. If we can trace the path between $\widehat{V}(0,0)$ and $\widehat{V}(T,N)$, we can establish the optimal strategy. The problem of finding an optimal strategy thus maps to the problem of finding a path between $\widehat{V}(0,0)$ and $\widehat{V}(T,N)$. The last iteration of Algorithm~\ref{alg:baseline} can be regarded as the procedure of finding such path by tracing from the leaf $\widehat{V}(T,N)$ back to the root $\widehat{V}(0,0)$.
 
\subsection{Complexity Reduction via Preprocessing}

As an optimization to further reduce the complexity of Algorithm~\ref{alg:baseline}, we add a preprocessing phase identifying the resources that are guaranteed to be included (excluded, respectively) in any optimal strategy $\pi^*$. The preprocessing allows to reduce the size of the problem to be solved by Algorithm~\ref{alg:baseline}. Our preprocessing phase hinges on the structural properties of $\pi^*$ given in Lemma~\ref{lemma:aux3}, which hinges on another structural property given in Lemma~\ref{lemma:aux2} and Definition~\ref{def:dom}.

\begin{definition}[Dominance]
\label{def:dom}
Given a pair of resources $i$ and $j$, we say that $i$ dominates $j$ in $q$ (in $c$, respectively) if $q_i\le q_j$ ($c_i\le c_j$). We say that $i$ dominates $j$ if $i$ dominates $j$ in both $q$ and $c$. The dominance is said to be strict if at least one inequality holds strictly.
\end{definition}

It can be noted that dominance defined above is a partial order, and that dominance in $q$ and in $c$ are total orders.

\begin{lemma}
\label{lemma:aux2}
Given any strategy $\pi$, if there exists a resource $\pi_k\in\pi$ dominated by another resource $\pi_k'\notin\pi$, it holds that $V(\pi)\ge V(\pi')$, where $\pi'$ denotes the strategy by replacing $\pi_k$ by $\pi_k'$ and keeping the other resources and their order as in $\pi$, i.e., $\pi'_j=\pi_j$ for $j\ne k$.
\end{lemma}

Lemma~\ref{lemma:aux2} is intuitive to understand in the sense that if resource $\pi_k'$ is superior to $\pi_k$ in both reward and cost, the user should choose $\pi_k'$ over $\pi_k$. Hence, $\pi^*$ never contains any resource dominated by another resource outside $\pi^*$. This intuition is further formalized and generalized in Lemma~\ref{lemma:aux3}.

\begin{lemma}
We sort the resources increasingly by $q_i$ and $c_i$, respectively. Let $y_i$ and $z_i$ denote the rank of resource $i$ based on the sorting of $q_i$ and $c_i$, respectively. The following two properties hold.
\begin{enumerate}
\item If $y_i+z_i\le T+1$, then $i\in\pi^*$.
\item If $y_i+z_i\ge N+T+1$, then $i\notin\pi^*$.
\end{enumerate}
\label{lemma:aux3}
\end{lemma}

Armed with Lemma~\ref{lemma:aux3}, we can develop a preprocessing procedure identifying all the resources satisfying the properties in Lemma~\ref{lemma:aux3}. We can then safely remove all resources identified by the preprocessing procedure, and only solve the remaining problem by Algorithm~\ref{alg:baseline} with reduced input size before reconstructing the optimal strategy by Lemma~\ref{lemma:optimal_property}. The preprocessing is straightforward to implement, whose pseudo-code is omitted here. The complexity of preprocessing is dominated by the sorting operation, which sums up to $O(N\log N)$ using heap-based implementation~\cite{Forsythe64}. 


To gain more quantitative insight on the benefit of preprocessing, we consider a system setting where  $q_i$ and $c_i$ for each resource $i$ are independently and randomly ranked in $[1,N]$. We can derive the percentage of resources removed by preprocessing. To this end, for any $i\in {\cal N}$, we have
\begin{align*}
    \Pr[y_i+z_i\le T+1]&=\sum_{y=1}^T\Pr[z_i\le T+1-y]\cdot \Pr[y_i=y] \\
    &=\sum_{y=1}^T\frac{T+1-y}{N}\cdot \frac{1}{N} =\frac{T(T+1)}{2N^2}.
\end{align*}
Symmetrically, we have 
\begin{align*}
    \Pr[y_i+z_i\ge N+T+1]&=\frac{(N-T)(N-T+1)}{2N^2}.
\end{align*}
Denote $\eta$ the percentage of resources removed by preprocessing in average and let $\beta\triangleq T/N$, asymptotically we have 
\begin{align*}
\eta &=\Pr[y_i+z_i\le T+1]+\Pr[y_i+z_i\ge N+T+1] \\
&\simeq \frac{\beta^2+(1-\beta)^2}{2}.
\end{align*}
Algebraically we have $0.25\le \eta\le 0.5$. The preprocessing procedure can thus reduce up to half of the total resources in the best case, thus halving the problem size passed to Algorithm~\ref{alg:baseline}. Even in the worst case where $T=N/2$ (corresponding to $\beta=0.5$), it can still filter out $25$\% resources.


\section{The Heterogeneous Case}
\label{sec:heterogeneous}

We proceed to the generic case where the access time is heterogeneous among resources. We develop our analysis by first demonstrating the hardness of the problem and then investigating the greedy strategy, followed by the development of a set of approximation algorithms approaching the minimal pseudo-cost with polynomial time and space complexity.

\begin{theorem}[\textsf{NP}-hardness of sequential resource access]
The sequential resource access problem in Definition~\ref{def:pb} is \textsf{NP}-hard.
\label{theorem:np}
\end{theorem}

The proof of Theorem~\ref{theorem:np}, detailed in~\cite{techreport}, consists of relating the sequential resource access problem to the $0-1$ Knapsack problem which is known to be \textsf{NP}-hard~\cite{Williamson2011}. 
Given the \textsf{NP}-hardness of our problem, we explore two directions in the sequel analysis. The first is to develop specific strategies that perform optimally under certain conditions. The second is to design efficient approximation algorithms with bounded efficiency loss to the optimal utility. By efficient, we mean the algorithm has polynomial complexity in both time and space.

\subsection{The Greedy Strategy}

We consider the greedy strategy defined in Definition~\ref{def:greedy}, i.e., accessing the resources in the decreasing order of $p_i/c_i$ until success, or reaching time $T$. Theorem~\ref{theorem:optimality_greedy_heter} establishes the sufficient conditions for the optimality of the greedy strategy. 

\begin{theorem}[Optimality condition of greedy strategy]
If $c_i< c_{i+1}$, $d_i\le d_{i+1}$ and $\frac{p_i-p_{i+1}}{c_i-c_{i+1}}< 1$ hold for $1\le i\le N-1$, the greedy strategy is the only optimal strategy.
\label{theorem:optimality_greedy_heter}
\end{theorem}

Theorem~\ref{theorem:optimality_greedy_heter} essentially demonstrates that the greedy strategy is optimal if (1) resources with higher access cost also incur higher delay, e.g., in the scenarios where the access cost is positively correlated or even proportional to the access delay, and (2) $p_i$ does not increase as much as $c_i$ among resources. 

More generically, if the access delay and cost are not positively correlated, the greedy strategy may not be optimal, motivating our following analysis on the most generic case.

\subsection{Approximation Algorithm Based on Dynamic Programming}

In contrast to the homogeneous case, our sequential resource access problem in the heterogeneous case is \textsf{NP}-hard. We thus concentrate on developing approximation algorithms achieving near-optimal performance. We first formalize in Definition~\ref{def:appro} the way we approximate an optimal strategy.

\begin{definition}[$(\epsilon,\delta)$-optimality]
A strategy $\pi$ is called $(\epsilon,\delta)$-optimal if $V(\pi)\le (1+\epsilon)V(\pi^*)$ and $d(\pi)\le (1+\delta)T$.
\label{def:appro}
\end{definition}

By Definition~\ref{def:appro}, we allow the total access time to slightly exceed the given time constraint $T$.  Our formulation makes sense if the delay constraint is not strictly inviolable so as to tolerate certain ``overflow''. The quantity of such overflow can be controlled by the parameter $\delta$. When $\epsilon=0$, $(0,\delta)$-optimality degenerates to the standard optimality with a $\delta$ relaxed constraint. When $\delta=0$, $(\epsilon,0)$-optimality degenerates to the classic $\epsilon$-optimality.
The focus of our work is the development and analysis of two polynomial algorithms outputting an $(\epsilon,\delta)$-optimal strategy. 


Our first approximation algorithm, presented in this subsection, extends from the dynamic programming approach in the homogeneous case. The core idea is to discretize the access delay to $O(1/\delta)$ values by scaling and rounding each $d_i$. This discretization step allows us to adapt the dynamic programming approach to the generic heterogeneous case. Specifically, our algorithm runs in two steps. To make the analysis concise without losing generality, we normalize $d_{min}\triangleq\min_{i\in{\cal N}} d_i=1$.

\textbf{Step 1: discretization}. Given $\delta > 0$, we set $\lambda \triangleq \lceil 1/\delta\rceil$ as a scaling parameter\footnote{In practice, $\delta$ is usually very small. We can thus conveniently approximate $\lambda$ as $1/\delta$ and treat it as a large integer.}, further replace each $d_i$ by $\widehat{d}_i\triangleq \left\lfloor d_i\lambda\right\rfloor/\lambda$, i.e., we round down the fractional part of $d_i$, $d_i-\lfloor d_i\rfloor$, to the closest fraction of the form $a/\lambda$ with $a<\lambda$ being an integer. 

\textbf{Step 2: dynamic programming}. Let $\widehat{T}\triangleq \left\lceil T/\delta\right\rceil$. For any pair of integers $(t,l)$, $0\le t\le \widehat{T}$, $1\le l\le N$, let $\widehat{V}(t,l)$ denote the minimal expected cost achievable by accessing a subset of resources among the first $l$ resources with a total delay at most $\delta\widehat{T}$, i.e., $\widehat{V}(t,l)\triangleq \min_{\pi\in \Pi, \widehat{d}(\pi)\le t\delta} V(\pi)$ where $\widehat{d}(\pi)\triangleq \sum_{i=1}^{|\pi|} \widehat{d}_{\pi_i}$. By extending the results in the homogeneous case, we can establish the recursive property concerning $\widehat{V}(t,l)$.
\begin{multline}
\widehat{V}(t,l)=\min\{c_l+q_l\cdot \widehat{V}(t-\widehat{d}_l\lambda,l-1),\widehat{V}(t,l-1)\}, \\
1\le t\le \widehat{T}, \ 1\le l\le N.
\label{eq:dp2}
\end{multline}

The pseudo-code of our approximation algorithm is described formally in Algorithm~\ref{alg:heter}. The only notable difference compared to the homogeneous case is the way how $\widehat{\pi}$ is established. Specifically, since we do not know the number of resources accessed in $\widehat{\pi}$, we reconstruct $\widehat{\pi}$ from the last element backward to the first and then reverse $\widehat{\pi}$, as depicted in the last line of Algorithm~\ref{alg:heter}, where \textsc{Reverse}($\widehat{\pi}$) denotes the operation of reversing $\widehat{\pi}$. 
Both the time and space complexity of Algorithm~\ref{alg:heter} is $O(N\widehat{T})$, i.e., $O(NT/\delta)$. Theorem~\ref{theorem:heter} establishes the $(0,2\delta)$-optimality of Algorithm~\ref{alg:heter}.

\begin{algorithm}
\caption{Approximation algorithm: heterogeneous case}
\label{alg:heter}
\begin{algorithmic}[1]
\State \textbf{Input}: $T$, $\{q_i,c_i\}_{1\le i \le N}$
\State \textbf{Output}: a $(0,2\delta)$-optimal strategy $\widehat{\pi}=\{\widehat{\pi}_i\}_{1\le i\le |\widehat{\pi}|}$

\vspace{0.5em}
\State Sort resources by decreasing $p_i/c_i$
\State $\lambda \leftarrow \lceil1/\delta\rceil$, $\widehat{T} \leftarrow \left\lceil T/\delta\right\rceil$
\For{$i=1$ \textbf{to} $N$}
    \State $\widehat{d}_i \leftarrow \left\lfloor d_i\lambda\right\rfloor/\lambda$
\EndFor

\For{$l=0$ \textbf{to} $N$, $t=0$ \textbf{to} $\widehat{T}$}
    \State $\widehat{V}(t,l)\leftarrow 1$
\EndFor

\vspace{0.5em}
\For{$l=1$ \textbf{to} $N$, $t=\widehat{d}_l\lambda$ \textbf{to} $\widehat{T}$, }
    \If{$c_l+q_l \widehat{V}(t-\widehat{d}_l\lambda,l-1)\le \widehat{V}(t,l-1)$}
        \State $\widehat{V}(t,l)\leftarrow c_l+q_l \widehat{V}(t-\widehat{d}_l\lambda,l-1)$
    \Else
        \State $\widehat{V}(t,l)\leftarrow \widehat{V}(t,l-1)$
    \EndIf
\EndFor

\vspace{0.5em}
\State $i\leftarrow 1$
\For{$l=N$ \textbf{to} $1$}
    \If{$c_l+q_l \widehat{V}(t-\widehat{d}_l\lambda,l-1)\le \widehat{V}(t,l-1)$}
        \State $\widehat{\pi}_i\leftarrow l$, $i\leftarrow i+1$, $t\leftarrow t-\widehat{d}_l\lambda$        
    \EndIf
\EndFor

\vspace{0.5em}
\State \textsc{Reverse}($\widehat{\pi}$)
\end{algorithmic}
\end{algorithm}

\begin{theorem}
Algorithm~\ref{alg:heter} outputs a $(0,2\delta)$-optimal strategy.
\label{theorem:heter}
\end{theorem}

To further reduce the complexity of Algorithm~\ref{alg:heter}, we can extend the preprocessing phase developed in the homogeneous case. It suffices to modify the definition of dominance by taking into account the heterogeneous access delay as below. Given a pair of resources $i$ and $j$, we say that $i$ dominates $j$ in $q$ (in $c$, $d$, respectively) if $q_i\le q_j$ ($c_i\le c_j$, $d_i\le d_j$). We say that $i$ dominates $j$ if $i$ dominates $j$ in $q$, $c$ and $d$. 

\subsection{Improved Approximation Algorithm}

By examining Algorithm~\ref{alg:heter}, we observe that the 2-dimensional table $\widehat{V}$ is not an efficient data structure for our problem, as some of the entries are not needed to
determine an optimal strategy. Motivated by this observation, we develop an improved algorithm, whose key idea is exposed as follows. 


In addition to discretize $d_i$, we also discretize the cost $V(\pi)$ for any strategy $\pi$ by setting $\mu \triangleq \lceil n^*/\epsilon c_{min}\rceil$ and replacing $V(\pi)$ by $\widehat{V}(\pi)\triangleq \left\lfloor V(\pi)\mu\right\rfloor/\mu$, where  $n^*\triangleq|\pi^*|$ denotes the number of resources accessed in $\pi^*$.\footnote{Without introducing much ambiguity, we use the same notation $\widehat{V}$ as in Algorithm~\ref{alg:heter}, as they essentially play similar roles in both algorithms.} Recall that $d_{min}=1$, we can loosely upper-bound $n^*$ by $T$ and set $\mu=T/\epsilon c_{min}$.

We then generate a list $\Gamma$ of all feasible pairs of $(\widehat{d}(\pi), \widehat{V}(\pi))$ ($\pi\in\Pi$) where $\widehat{d}(\pi)\le T$. Technically, $\Gamma$ can be generated in $N$ iterations. Initially we put $(0,1)$ in $\Gamma$. At iteration $l$, from each pair $(\widehat{d}(\pi),\widehat{V}(\pi))$, we generate another pair $\left(\widehat{d}(\pi)+\widehat{d}_l, 
\left\lfloor(c_l+q_l\widehat{V}(\pi))\mu\right\rfloor/\mu\right)$ if $\widehat{d}(\pi)+\widehat{d}_l\le T$ based on the recursive property of $\widehat{d}(\pi)$ established in~\eqref{eq:dp2}, and add the new pair to $\Gamma$ if it does not duplicate any existing pair. By doing so, at the end of iteration $l$, each pair in $\Gamma$ represents a strategy of accessing a subset of resources among the first $l$ resources, whose total access time is upper-bounded by $T$, and inversely, each such strategy is represented by a pair. Once $\Gamma$ is established, we return the strategy $\widehat{\pi}$ corresponding to the pair with minimal discretized cost $\widehat{V}(\widehat{\pi})$ as an approximate optimal strategy. 

The above algorithm can be further improved by noting that not all pairs in $\Gamma$ are needed to derive an approximate optimal strategy. In fact we can safely remove a pair $(\widehat{d}(\pi), \widehat{V}(\pi))$ if there exists another pair $(\widehat{d}(\pi'), \widehat{V}(\pi'))$ in $\Gamma$ such that the former is dominated by the latter in the sense $\widehat{d}(\pi)\ge \widehat{d}(\pi')$ and $\widehat{V}(\pi)\ge \widehat{V}(\pi')$. After eliminating all dominated pairs, each
remaining pair $(\widehat{d}(\pi), \widehat{V}(\pi))$ is Pareto-optimal by satisfying the following conditions at the end of iteration $l$: (1) $\widehat{d}(\pi)$ is the smallest access delay with access cost $\widehat{V}(\pi)$ if only the first $l$ resources are allowed to be accessed, (2) $\widehat{V}(\pi)$ is the smallest access cost with access delay $\widehat{d}(\pi)$ if only the first $l$ resources are allowed to be accessed.

Based on the above idea, we can now modify our improved algorithm by adding a pruning phase at the end of each iteration. We sort the pairs in $\Gamma$ in strictly increasing order of $\widehat{d}$ and in strictly decreasing order of $\widehat{V}$ due to elimination of dominated pairs. In iteration $l$, we produce a new list $\Gamma'$ as follows: for each pair $(\widehat{d}(\pi), \widehat{V}(\pi))\in \Gamma$, we add a pair $\left(\widehat{d}(\pi)+\widehat{d}_l, 
\left\lfloor(c_l+q_l\widehat{V}(\pi))\mu\right\rfloor/\mu\right)$ to $\Gamma'$ if $\widehat{d}(\pi)+\widehat{d}_l\le T$. Since the pairs in $\Gamma$ is in increasing order of $\widehat{d}(\pi)$, the process of establishing $\Gamma'$ can be terminated whenever a pair $(\widehat{d}(\pi), \widehat{V}(\pi))$ in $\Gamma$ is reached, for which $\widehat{d}(\pi)+\widehat{d}_l(\pi)>T$. Once $\Gamma'$ is established, we merge $\Gamma'$ to $\Gamma$ by removing dominated pairs. This is easily accomplished given the strict ordering of $\widehat{d}(\pi)$ and $\widehat{V}(\pi)$ in the list. At the end of the last iteration (iteration $N$), the last pair in $\Gamma$ gives the cost and the corresponding access delay for an approximated optimal solution we look for.
To reconstruct the set of resources to access, we can simply add a pointer to each pair, pointing to the parent pair from which the current pair is generated. Via these pointers we can trace back from the last pair in $\Gamma$ to the first pair $(0,1)$, and reconstruct the set of resources to access.  

The pseudo-code of the improved approximation algorithm is described in Algorithm~\ref{alg:heter2}. The algorithm is mainly composed of two iterations, except that of initialization. The first iteration (line~\ref{line:for_start} to line~\ref{line:for_end}) establishes $\Gamma$ as described above. The second iteration (the \textbf{while} loop) reconstructs the entire set of resources to access. In our algorithm, the following elementary functions are used. They are straightforward to implement and are omitted here for briefness.
\begin{itemize}
    \item \textsc{Insert}($e$, $\gamma$) inserts the element $e$ in the list $\gamma$.
    \item \textsc{Merge}($\Gamma$, $\Gamma'$) merges $\Gamma'$ into $\Gamma$, removing dominated pairs and sorting the resulting list. Given the structure of $\Gamma$, where the pairs are in increasing order of $\widehat{d}(\pi)$ and decreasing order of $\widehat{V}(\pi)$, \textsc{Merge}($\Gamma$, $\Gamma'$) can be implemented in linear time to the size of $\Gamma$ and $\Gamma'$.
    \item \textsc{Last}($\Gamma$) returns the last element in the list $\Gamma$.
    \item \textsc{Find}($\gamma$, $ptr$) returns the resource such that by accessing it, the pair (access delay, cost) changes from the values in the pair pointed by $ptr$ to the values in $\gamma$. This can be achieved via a standard search in $O(\log N)$ time if we provide a sorted list of resources based on access delay.
\end{itemize}

\begin{algorithm}
\caption{Improved approximation algorithm}
\label{alg:heter2}
\begin{algorithmic}[1]
\State \textbf{Input}: $T$, $\{q_i,c_i\}_{1\le i \le N}$
\State \textbf{Output}: an $(\epsilon,\delta)$-optimal strategy $\widehat{\pi}=\{\widehat{\pi}_i\}_{1\le i\le |\widehat{\pi}|}$

\vspace{0.5em}
\State Sort resources by decreasing $p_i/c_i$
\State $\lambda \leftarrow \lceil1/\delta\rceil$, $\mu \leftarrow \lceil1/\epsilon c_{min}T\rceil$
\For{$i=1$ \textbf{to} $N$}
    \State $\widehat{d}_i \leftarrow \left\lfloor d_i\lambda\right\rfloor/\lambda$ \label{line:rounding}
\EndFor
\State $\Gamma\leftarrow (0,1,\emptyset)$

\vspace{0.5em}
\For{$l=1$ \textbf{to} $N$} \label{line:for_start}
    \State $\Gamma'\leftarrow \emptyset$
    \For{\textbf{each} $\gamma=(\widehat{d}(\pi),\widehat{V}(\pi),ptr(\pi))\in \Gamma$}
        \If{$\widehat{d}(\pi)+\widehat{d}_l>T$}
            \State \textbf{break}
        \Else   
            \State $ptr\leftarrow \&\gamma$
            \State \textsc{Insert}($\left(\widehat{d}(\pi)+\widehat{d}_l, 
\left\lfloor(c_l+q_l\widehat{V}(\pi))\mu\right\rfloor/\mu,\right.$ $\left.ptr\right)$, $\Gamma'$)
        \EndIf
    \EndFor
    \State \textsc{Merge}($\Gamma$, $\Gamma'$)\label{line:for_end}
\EndFor

\vspace{0.5em}
\State $i\leftarrow 1$, $\gamma\triangleq(\widehat{d}(\pi),\widehat{V}(\pi),ptr(\pi))\leftarrow$ \textsc{Last}($\Gamma$)
\While{$ptr(\pi)\ne\emptyset$}
    \State $\widehat{\pi}_i\leftarrow$ \textsc{Find}($\gamma$, $ptr(\pi)$) 
    \State $\gamma\leftarrow ptr(\pi)$, $i\leftarrow i+1$ 
\EndWhile

\vspace{0.5em}
\State \textsc{Reverse}($\widehat{\pi}$)
\end{algorithmic}
\end{algorithm}

We next analyze the time and space complexity of Algorithm~\ref{alg:heter2}. In this regard, we first derive the maximal size of $\Gamma$ and $\Gamma'$. Note that (1) the finest granularity of the discretized access delay is $1/\lambda$, and the total delay is upper bounded by $T$, (2) the finest granularity of the discretized cost $V(\pi)$ is $1/\mu$, and $V(\pi)\le 1$ from Lemma~\ref{lemma:v}, and (3) all the dominated pairs are removed from $\Gamma$, which implicates that (3.1) for any discretized access delay $\widehat{d}$ there is at most one pair with delay $\widehat{d}$, (3.2) for any discretized cost $\widehat{v}$ there is at most one pair with cost $\widehat{v}$. Therefore, the maximal size of $\Gamma$ is $\min\{T/\lambda,1/\delta\}$, which, by injecting $\lambda$ and $\mu$, sums up to $\min\{|\pi^*|T/\epsilon c_{min},1/\delta\}$, where $|\pi^*|$ can be upper-bounded by $T$. By the construction of $\Gamma'$ in Algorithm~\ref{alg:heter2}, where each pair in $\Gamma$ generates at most one pair in $\Gamma'$, we can upper-bound the size of $\Gamma'$ also by $\min\{|\pi^*|T/\epsilon c_{min},1/\delta\}$. The overall time complexity of Algorithm~\ref{alg:heter2}, dominated by the second \textbf{for} loop, sums up to $O\left(N\min\{|\pi^*|T/\epsilon c_{min},1/\delta\}\right)$. Algorithm~\ref{alg:heter2} also needs to stock $\Gamma$ and $\Gamma'$. Hence the space complexity sums up to $O\left(\min\{T^2/\epsilon c_{min},1/\delta\}\right)$, if we do not count the space required to store the parameters. We remark that this is an order-of-magnitude space gain compared to Algorithm~\ref{alg:heter}.


Theorem~\ref{theorem:heter2} establishes the $(\epsilon,\delta)$-optimality of Algorithm~\ref{alg:heter2}.


\begin{theorem}
Algorithm~\ref{alg:heter2} gives an $(\epsilon,\delta)$-optimal strategy.
\label{theorem:heter2}
\end{theorem}

\begin{table*}[htbp]
\begin{tabular}{p{7cm} p{2.7cm} p{3.4cm} p{3cm}}
\toprule
\textbf{Algorithm}                                                                                 & \textbf{Optimality}                                                                                         & \textbf{Time complexity} & \textbf{Space complexity} \\ \toprule
Greedy algorithm                                                                                   & optimal under certain condition                               & $O(T)$                                                              & $O(T)$                                                               \\ 
\hline
DP-based algorithm (Algo.~\ref{alg:baseline}): homogeneous case        & optimal                                                                                                     & $O(NT)$                                                             & $O(NT)$                                                              \\ \hline
Approximation algorithm  (Algo.~\ref{alg:heter})                                                                  & $(0,2\delta)$-optimal                                                                                        & $O(NT/\delta)$                                                       & $O(NT/\delta)$                                                        \\ \hline
Improved approximation  algorithm (Algo.~\ref{alg:heter2})             & $(\epsilon,\delta)$-optimal                                                                                 & $O\left(N\min\left\{\frac{T^2}{\epsilon c_{min}},\frac{1}{\delta}\right\}\right)$              & $O\left(\min\left\{\frac{T^2}{\epsilon c_{min}},\frac{1}{\delta}\right\}\right)$                \\ \hline
\end{tabular}
\vspace{-0.0cm}
\caption{Comparison of different algorithms developed in our work}
\label{tab:comparison}
\end{table*}

We conclude this section by summarizing and comparing the different algorithms we develop in Table~\ref{tab:comparison}. For the time and space complexity, the time and space required to sort and store the parameters of each resource (e.g., $p_i$, $c_i$, $d_i$) are not taken into account. The minimal access delay $d_{min}$ is normalized to $1$.

\section{Numerical Analysis}
\label{sec:simu}

In this section, we conduct numerical analysis to evaluate the performance of our approximation algorithm. We
evaluate our algorithm against the following three strategies:
\begin{itemize}
    \item \textbf{Randomized strategy}: the user accesses a randomly chosen resource each time until success or reaching $T$;
    \item \textbf{Greedy-$p$}: the user accesses the resources in decreasing order of $p_i$ until success or the end of $T$;
    \item \textbf{Greedy-$c$}: the user accesses the resources in decreasing order of $c_i$ until success or the end of $T$;
    \item \textbf{Greedy-$p/c$}: the user accesses the resources in decreasing order of $p_i/c_i$ until success or the end of $T$; this strategy is the greedy strategy analyzed in Section~\ref{sec:homogeneous} and~\ref{sec:heterogeneous}.
\end{itemize}

Specifically, we simulate a system consisting of $N=20$ resources, where each $p_i$ is randomly chosen from $[0,1]$, $c_i$ randomly from $[0,p_i]$, $d_i=1$ for the homogeneous case and randomly chosen from $[0,2]$ in the heterogeneous case, $T$ varies from $4$ to $16$. We take the randomized strategy as the baseline and trace the ratio between the utility of the other simulated strategies and the utility of the randomized strategy, denoted by $\Psi$. In other words, $\Psi$ quantifies the performance gain of the simulated strategy over the randomized baseline strategy. For our algorithm, we run Algorithm~\ref{alg:heter2} with both $\delta$ and $\epsilon$ set to $0.01$. Figure~\ref{fig:homo} and~\ref{fig:heter} traces the average of $\Psi$ over $10000$ runs for the homogeneous and heterogeneous cases, respectively. We make the following observations.

Our algorithm performs constantly and significantly better than all the other strategies. This is coherent to the theoretical results as our algorithm is proved to approach the system optimum. 

Among the other strategies, none outperforms the others in all the simulated scenarios. We can only observe that, in the average sense, greedy-$p$ and greedy-$p/c$ performs better with small $T$, while greedy-$c$ performs better with large $T$. However, this performance trade-off depends on the system parameters.  

The performance gain of our algorithm is more significant where approximately half of the resources can be selected (i.e.,  $T=8$ and $12$). This is also the least tractable cases, since when the number of selected resources approaches $N/2$, the number of choices are maximized, hence the optimum strategy may significantly outperforms a heuristic one, as demonstrated by our results.

\begin{figure}[ht]
\centering
\subfigure{%
\includegraphics[width=0.45\textwidth]{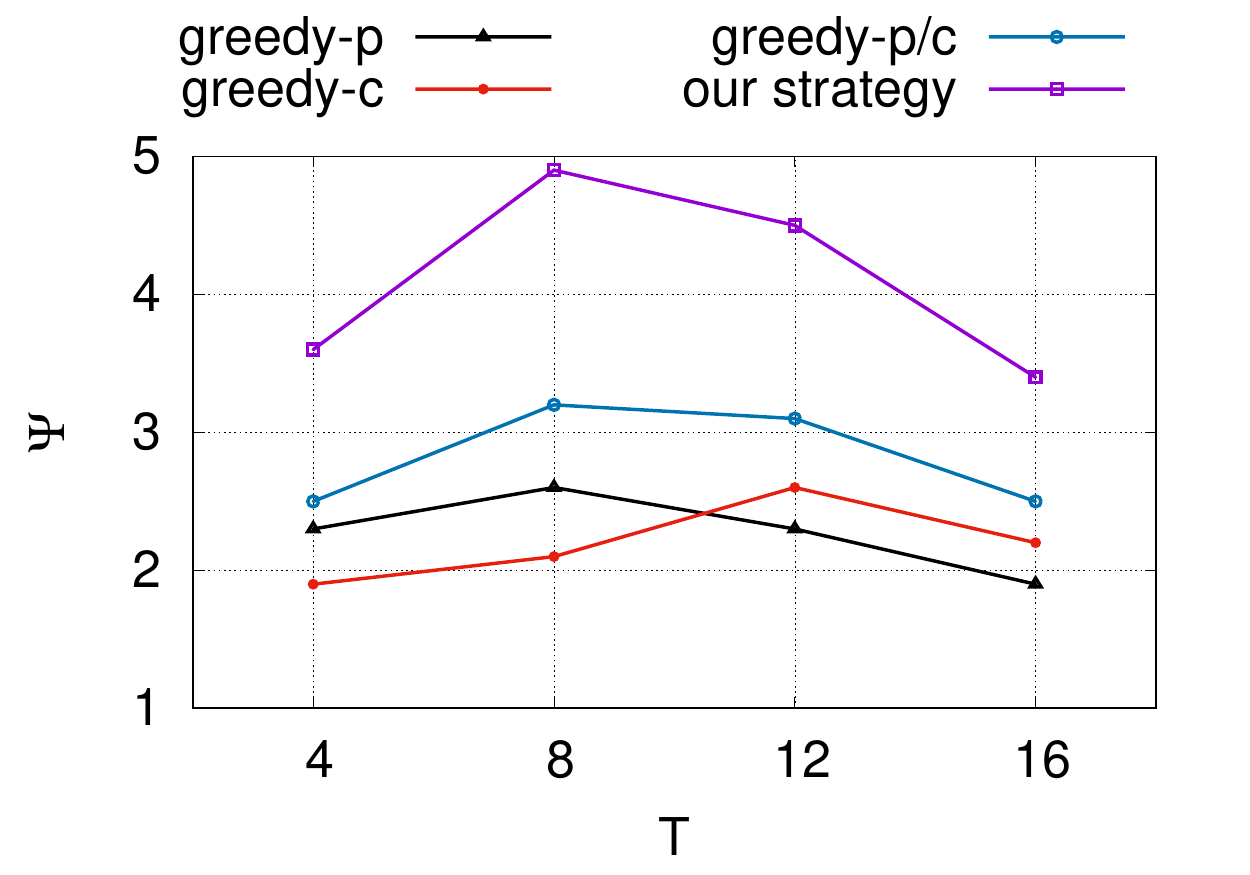}
\label{fig:homo}}
\subfigure{%
\includegraphics[width=0.45\textwidth]{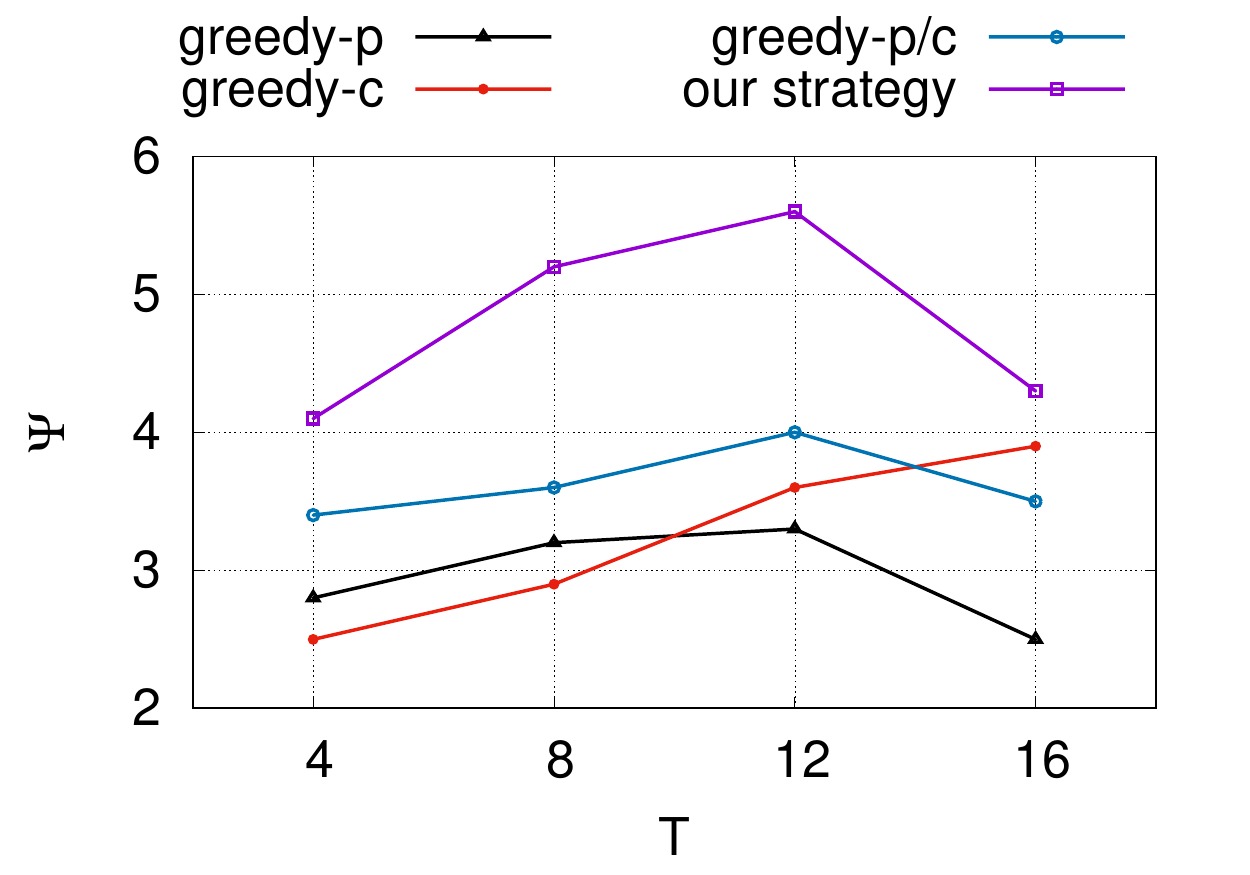}
\label{fig:heter}}
\vspace{-0.2cm}
\caption{Performance of different access strategies: heterogeneous (top) and heterogeneous (bottom) cases.}
\label{fig:delay}
\end{figure}

\section{Related Work}
\label{sec:related_work}

The problem we tackle in this paper is related to the sequential decision making problem, among which perhaps the most pertinent example to our context is the classic optimal stopping problem~\cite{Shiryaev2011}. Stemmed from the famous secretary problem dating back to the late 1950's, the theory of optimal stopping has received a lot of research attention by constituting today a field of study in applied mathematics and statistics~\cite{10.2307/2245639}. Briefly speaking, the theory of optimal stopping is concerned with the problem of deciding when to stop based on sequentially observed random variables in order to maximize an expected payoff. Interested readers may refer to~\cite{stopping2,Shiryaev2011} for a detailed treatment. Compared to the standard optimal stopping problem where only a binary action, stopping or proceeding to observe, needs to be taken upon each
observation, our problem has a more combinatorial flavor as the decision maker needs to choose which resources to probe and in what sequence. In contrast, our problem does not have the stochastic component as in the optimal stopping problem since the system parameters are known \textit{a priori}.   

From an application point of view, since the last decade, there has been a surge of interest in using optimal stopping theory in emerging communication and networking paradigms~\cite{icn-forward,9014081,9052123,5233772,5992834,5738213}. These studies consider a variety of scenarios where a user seeks a sequential strategy that probes a subset of resources with potentially unknown distribution of states so as to maximize her utility. The considered problem can be cast to the optimal stopping problem by regarding resource states as the observable variable. Two variants have been investigated depending on whether the user is allowed to use a previously probed resource or not. Under certain system setting, the optimal probing policy has a threshold structure and can be derived in a subset of special cases (e.g., with two resources)~\cite{5233772}. Compared to these works addressing concrete application examples with particular constraints and specificity therein, our work take a more generic view without sacrificing the complexity and intrinsic structure of the problem. Therefore, we believe that our analysis, despite being generic, can serve as a complementary line of research, and provide valuable insights in a variety of resource access problems in emerging networking and computing systems.



\section{Conclusion and Perspective}
\label{sec:conclusion}

We have formulated and analyzed a generic resource access problem, where a user, having access to a set of resources with heterogeneous success probability, access cost and latency, seeks an optimal strategy to achieve the best trade-off between limiting the access cost and maximizing the overall success probability. We proved that the problem of finding an optimal strategy is \textsf{NP}-hard. Given its \textsf{NP}-hardness, we presented a greedy strategy that can be implemented in linear time and mathematically establish sufficient conditions for its optimality. We then developed a series of polynomial-time approximation algorithms approaching an optimum solution in the sense of $(\epsilon,\delta)$-optimality.

A natural generalization of our current work, as mentioned in the Introduction, is to extend our algorithmic framework to investigate the case with correlated resources. This is an important facet of the problem, and will provide valuable insight on how to exploit the correlation among resources to improve the performance. Another direction is to investigate the scenario, where the user is allowed to access up to $C$ resources simultaneously each time, and can perform multiple rounds of queries, as long as the total delay does not exceed $T$. This direction adds more combinatorial flavor to the problem.

\bibliographystyle{abbrv}
\bibliography{bibfile}

\end{document}